\documentclass[fleqn]{article}
\usepackage{amsmath}
\usepackage{amsthm}
\usepackage{amssymb}
\usepackage{enumerate}
\usepackage[all]{xy}

\theoremstyle{plain}
\newtheorem{lem}{Lemma}
\newtheorem{thm}[lem]{Theorem}

\DeclareMathOperator{\Iso}{\mathbf{Iso}}
\DeclareMathOperator{\Isoc}{\mathbf{Isoc}}
\newcommand{\ilim}{\mathop{\varprojlim}\limits}

\usepackage[pdftex]{hyperref}

\begin{document}
\begin{center}
\Large
{\bf An algebraically independent generating set of the algebra of local unitary invariants}
\end{center}
\vspace*{-.1cm}
\begin{center}
P\'eter Vrana
\end{center}
\vspace*{-.4cm} \normalsize
\begin{center}
Department of Theoretical Physics, Institute of Physics, Budapest University of\\ Technology and Economics, H-1111 Budapest, Hungary

\vspace*{.2cm}
(\today)
\end{center}

\begin{abstract}
We show that the inverse limit of the graded algebras of local unitary invariant
polynomials of finite dimensional $k$-partite quantum systems is free, and give
an algebraically independent generating set. The number of degree $2d$ invariants
in the generating set is equal to the number of conjugacy classes of index $d$
subgroups of a free group on $k-1$ generators.
\end{abstract}

\section{Introduction}

One of the approaches to the problem of understanding quantum entanglement is
to look for functions on the state space which are invariant under the action
of the local unitary (LU) group and enable us to distinguish between different
types of states. For a multipartite quantum system with distinguishable
subsystems, this group is the product of the unitary groups acting on the
Hilbert spaces of the individual subsystems.

The problem of separating the orbits can be reduced to finding the set of polynomial
invariants \cite{MW}, which form an algebra. Unfortunately, a description in terms
of generators and relations is available only in the case of some special Hilbert
space dimensions and particle numbers, in other cases, only partial results exist,
see e.g. \cite{Verstraete, Brylinski, LT, LTT}.

In \cite{HW} it was pointed out that the dimension of LU-invariant homogenous
polynomials with a fixed degree stabilizes as the dimensions of the Hilbert
spaces of the subsystems increase. Based on this observation, one can introduce
an algebra which can be thought of as gluing together the algebras of LU-invariant
polynomials of various finite dimensional quantum systems \cite{Vrana}, much like
one studies the algebra of symmetric polynomials independently of the number of
variables.

The outline of the paper is as follows. In section~\ref{sec:LUinv} we summarize
the construction of the inverse limit of the algebras of LU-invariant polynomials
over finite dimensional state spaces of quantum systems with a fixed number of
subsystems. We call this object the algebra of local unitary invariants \cite{Vrana}.

In section~\ref{sec:cover} we collect some facts about finite coverings of a
graph. In particular, we describe a bijection between conjugacy classes of
finite index subgroups of a free group, finite coverings of a certain graph
and orbits of tuples of permutations under simultaneous conjugation following
ref. \cite{Kwak}.

In section~\ref{sec:alggen} we prove that the algebra of local unitary invariants
is free by giving an algebraically independent generating set. Our proof makes
use of the invariants introduced in ref. \cite{HWW}.

Section~\ref{sec:conclusion} contains some concluding remarks, including the
interpretation of our result in the context of LU-invariants of mixed states.

\section{The algebra of local unitary invariants}\label{sec:LUinv}

Let $k\in\mathbb{N}$ and for every $k$-tuple $n=(n_1,\ldots,n_k)\in\mathbb{N}^k$
let us consider the complex Hilbert space $\mathcal{H}_n=\mathbb{C}^{n_1}\otimes\cdots\otimes\mathbb{C}^{n_k}$
describing the pure states of a composite system with $k$ distinguishable
subsystems. The group of local unitary transformations, $LU_n=U(n_1,\mathbb{C})\times\cdots\times U(n_k,\mathbb{C})$,
acts on $\mathcal{H}$ in the obvious way, i.e. regarding $\mathbb{C}^{n_i}$ as
the standard representation of $U(n_i,\mathbb{C})$.

Let $I_{k,n}$ denote the algebra of $LU_n$-invariant polynomial functions over
$\mathcal{H}_n$, regarded as a real vector space. Polynomial functions (with
respect to any fixed basis) are in bijection with elements in
$S(\mathcal{H}_n\oplus\mathcal{H}_n^{*})$, the symmetric algebra on
$\mathcal{H}_n\oplus\mathcal{H}_n^{*}$ on which an action of $LU_n$ is induced
and we have
\begin{equation}
I_{k,n}=S(\mathcal{H}_n\oplus\mathcal{H}_n^{*})^{LU_n}
\end{equation}
Note that in $I_{k,n}$ the polynomials are of the same degree in the coefficients
and their complex conjugates, therefore we find it convenient to use a grading
which is different from the usual one in a factor of two, and call homogenous
degree $m$ the polynomials which are of degree $m$ both in the coefficients and
their conjugates.

For $n\le n'\in\mathbb{N}^k$ with respect to the componentwise order, we have
the inclusion $\iota_{n,n'}:\mathcal{H}_n\hookrightarrow\mathcal{H}_{n'}$ which
is the tensor product of the usual inclusions $\mathbb{C}^{n_i}\hookrightarrow\mathbb{C}^{n'_i}$
sending an $n_i$-tuple to the first $n_i$ components. Similarly, we regard
$LU_n$ as a subgroup of $LU_{n'}$ which stabilizes the image of $\iota_{n,n'}$,
and thus $\iota_{n,n'}$ is an $LU_n$-equivariant linear map, inducing
a morphism of graded algebras $\varrho_{n,n'}:I_{k,n'}\to I_{k,n}$.

$((I_{k,n})_{n\in\mathbb{N}^k},(\varrho_{n,n'})_{n\le n'\in\mathbb{N}^k})$
is an inverse system of graded algebras, the inverse limit of which will
be denoted by $I_k$ and called the algebra of LU-invariants:
\begin{equation}
I_k:=\ilim_{n\in\mathbb{N}^k}I_{k,n}=\left\{(f_n)_{n\in\mathbb{N}^k}\in\prod_{n\in\mathbb{N}^k}I_{k,n}\Bigg|\forall n\le n':f_n=\varrho_{n,n'}f_{n'}\right\}
\end{equation}
Note that $I_{k,(n_1,\ldots,n_k)}$ is a quotient of $I_k$ and the restriction
of the quotient map to the subspace of elements of degree at most $\min\{n_1,\ldots,n_k\}$
is an isomorphism.

The dimension of the homogenous degree $m$ subspace of $I_k$ is given by
\begin{equation}\label{eq:stabdim}
d_{k,m}=\sum_{a\Vdash m}\left(\prod_{i=1}^{m}i^{a_i}a_i!\right)^{k-2}
\end{equation}
and the Hilbert series of $I_k$ is \cite{Vrana}
\begin{equation}
\sum_{m\ge 0}d_{k,m}t^m = \prod_{d\ge 1}(1-t^d)^{-u_d(F_{k-1})}
\end{equation}
where $u_d(F_{k-1})$ denotes the number of conjugacy classes of index $d$
subgroups of $F_{k-1}$, the free group on $k-1$ generators.

Our aim is to prove that $I_k$ is free, and the number of degree $d$ invariants
in an algebraically independent generating set equals the number of conjugacy
classes of index $d$ subgroups in the free group on $k-1$ generators, as
the Hibert series suggests.

\section{Graph coverings}\label{sec:cover}

Let $G=(V,E)$ be a connected graph with coloured and directed edges (possibly multiple edges
and/or loops). A graph $\tilde{G}=(\tilde{V},\tilde{E})$ together with a projection
$p:\tilde{G}\to G$ is said to be a covering of $G$ if $p_V:\tilde{V}\to V$ and
$p_E:\tilde{E}\to E$ are two surjections where the image of the head (tail) of an edge
is the head (tail) of its image, $p_E$ respects colours and such that the indegree
and outdegree of every vertex $\tilde{v}\in\tilde{V}$ is the same as that of
$p_V(\tilde{v})$ in each subgraph determined by the colours. A covering $p:\tilde{G}\to G$
is said to be finite if $|p_V^{-1}(v)|<\infty$ and $m$-fold if $|p_V^{-1}(v)|=m$
for all $v\in G$.

Two coverings $p_1:\tilde{G}_1\to G$ and $p_2:\tilde{G}_2\to G$ are said to
be isomorphic if there exists an isomorphism $\varphi:\tilde{G}_1\to\tilde{G}_2$
making the following diagram commute:
\begin{equation}
\xymatrix{\tilde{G}_1 \ar@{->>}[dr]_{p_1}\ar[rr]^{\varphi}  &  & \tilde{G}_2 \ar@{->>}[dl]^{p_2}  \\
          & G}
\end{equation}
The set of isomorphism classes of finite coverings of a graph $G$ will be denoted by
$\Iso(G)$, the set of isomorphism classes of $m$-fold coverings by $\Iso(G,m)$,
and the connected ones by $\Isoc(G)$ and $\Isoc(G,m)$, respectively.

Let $G$ be the graph with a single vertex and $k-1$ directed coloured loops.
It is well-known that its fundamental group $\pi_1(G)$ is the free group of rank
$k-1$, and a set of generators may be identified with the directed loops. There
exists a bijection between $\Isoc(G,m)$ and conjugacy classes of subgroups of
index $m$ of $\pi_1(G)$.

Let $S_m$ denote the group of bijections from $\{1,\ldots,m\}$ to itself. This
group acts on $S_m^{k-1}=S_m\times S_m\times\cdots\times S_m$ by simultaneous
conjugation:
\begin{equation}
\pi\cdot(\sigma_1,\ldots,\sigma_{k-1})=(\pi\sigma_1\pi^{-1},\ldots,\pi\sigma_{k-1}\pi^{-1})
\end{equation}
Let us denote the orbit of $(\sigma_1,\ldots,\sigma_{k-1})\in S_m^{k-1}$ by
\begin{equation}
[\sigma_1,\ldots,\sigma_{k-1}]=\{\pi\cdot(\sigma_1,\ldots,\sigma_{k-1})|\pi\in S_m\}
\end{equation}

To a (not necessarily connected) $m$-fold covering $\tilde{G}$ of $G$ we can associate
an element in
\begin{equation}
S_m^{k-1}/S_m=\{[\sigma_1,\ldots,\sigma_{k-1}]|\forall i:\sigma_i\in S_m\}
\end{equation}
as follows. Label the vertices of $\tilde{G}$ arbitrarily with the numbers
$\{1,\ldots,m\}$ using each label exactly once. Let $\sigma_i$ be the permutation
which sends $a$ to $b$ if there is a directed edge from $a$ to $b$ of colour $i$
in $\tilde{G}$. Note that this gives indeed a $k-1$-tuple of permutations as the
indegree and outdegree of every vertex in $\tilde{G}$ is $1$ in the subgraph
determined by any colour. As relabelling corresponds to simultaneous conjugation,
we have indeed a well-defined map $\Phi:\Iso(G,m)\to S_m^{k-1}/S_m$.

Now let $\tilde{G}_1$ and $\tilde{G}_2$ be two coverings of $G$ where $\tilde{G}_1$
is $m_1$-fold and $\tilde{G}_2$ is $m_2$-fold. The disjoint union $\tilde{G}:=\tilde{G}_1\sqcup\tilde{G}_2$
is then an $m_1+m_2$-fold covering of $G$. We would like to relate the orbits
of $k-1$-tuples of permutations of the three coverings. Let us choose the numbering
of the vertices of $\tilde{G}$ so that $\tilde{G}_1$ is labelled with $\{1,\ldots,m_1\}$
and $\tilde{G}_2$ is labelled with $\{m_1+1,\ldots,m_1+m_2\}$.

Let $(\sigma^{(j)}_1,\ldots,\sigma^{(j)}_{r-1})$ be the representative of the
orbit corresponding to $\tilde{G}_j$ and $(\sigma_1,\ldots,\sigma_{r-1})$ be
that of $\tilde{G}$ which can be read off from the above-chosen labelling
(after subtracting $m_1$ in the $j=2$ case). It is easy to see that for all
$1\le i\le k-1$ we have
\begin{equation}\label{eq:star}
\sigma_i(a)=\left\{\begin{array}{ll}
\sigma^{(1)}_i(a) & \textrm{if $a\le m-1$}  \\
\sigma^{(2)}_i(a-m_1)+m_1 & \textrm{if $a> m-1$}  \\
\end{array}\right.
\end{equation}
given by the usual homomorphism $S_{m_1}\times S_{m_2}\hookrightarrow S_{m_1+m_2}$.
This map clearly induces a map $\star:S_{m_1}^{k-1}/S_{m_1}\times S_{m_2}^{k-1}/S_{m_2}\to S_{m_1+m_2}^{k-1}/S_{m_1+m_2}$
on the orbits (we will use infix notation, i.e. the map sends $(a,b)\mapsto a\star b$).
One can see immediately that $\star$ turns the set $\bigsqcup_{m=1}^{\infty}S_m^{k-1}/S_m$
into a commutative semigroup. Also, $\Iso(G)$ can be equipped with a semigroup
structure induced by disjoint union, and $\Phi$ is an isomorphism.

\section{Algebraically independent generators of the algebra of LU-invariants}\label{sec:alggen}

To an orbit in $S_m^{k-1}/S_m$ we can associate an element in $I_k$ as follows.
It was shown in ref. \cite{Vrana} that every element of $I_k$ is represented in
some $I_{k,n}$. We will give a representative in $I_{k,n}$ where $n=(n_1,\ldots,n_k)\ge(m,\ldots,m)$ following
ref. \cite{HWW}. A vector in $\mathcal{H}_{n}$ is of the form
\begin{equation}
\psi=\sum_{i_1,\ldots,i_k}\psi_{i_1,\ldots,i_k}e_{i_1}\otimes\cdots\otimes e_{i_k}
\end{equation}
where in the sum $1\le i_j \le n_j$ for all $1\le j\le k$.

Let $(\sigma_1,\ldots,\sigma_{k-1})\in S_m^{k-1}$ be a representative. The value
of the associated polynomial on $\psi$ is
\begin{equation}\label{eq:f}
f_{[\sigma_1,\ldots,\sigma_{k-1}]}(\psi)=\sum_{i^{1}_{1},\ldots,i^{m}_{k}}\psi_{i^{1}_1,\ldots,i^{1}_k}\cdots\psi_{i^{m}_1,\ldots,i^{m}_k}\overline{\psi_{i^{\sigma_1(1)}_1,\ldots,i^{\sigma_{k-1}(1)}_{k-1},i^{1}_k}}\cdots\overline{\psi_{i^{\sigma_1(m)}_1,\ldots,i^{\sigma_{k-1}(m)}_{k-1},i^{m}_k}}
\end{equation}
where the sum is over all $k\cdot m$-tuples of integers where $1\le i^l_j \le n_j$
for all $1\le j\le k$ and $1\le l\le m$. Note that the expression defining
$f_{[\sigma_1,\ldots,\sigma_{k-1}]}$ is independent of the choice of the
representative, justifying the notation.

An important observation is the following:
\begin{lem}
Let $[\sigma^{(1)}_1,\ldots,\sigma^{(1)}_{k-1}]\in S_{m_1}^{k-1}/S_{m_1}$ and
$[\sigma^{(2)}_1,\ldots,\sigma^{(2)}_{k-1}]\in S_{m_2}^{k-1}/S_{m_2}$. Then
\begin{equation}
f_{[\sigma^{(1)}_1,\ldots,\sigma^{(1)}_{k-1}]\star[\sigma^{(2)}_1,\ldots,\sigma^{(2)}_{k-1}]}=
f_{[\sigma^{(1)}_1,\ldots,\sigma^{(1)}_{k-1}]}f_{[\sigma^{(2)}_1,\ldots,\sigma^{(2)}_{k-1}]}
\end{equation}
\end{lem}
\begin{proof}
Let the representative of the orbit
$[\sigma^{(1)}_1,\ldots,\sigma^{(1)}_{k-1}]\star[\sigma^{(2)}_1,\ldots,\sigma^{(2)}_{k-1}]$
given by eq. (\ref{eq:star}) be $(\sigma_1,\ldots,\sigma_{k-1})\in S_{m_1+m_2}^{k-1}$ and let us denote
$m_1+m_2$ by $m$. Then
\begin{equation}
\begin{split}
& \sum_{i^{1}_{1},\ldots,i^{m}_{k}}\psi_{i^{1}_1,\ldots,i^{1}_k}\cdots\psi_{i^{m}_1,\ldots,i^{m}_k}\overline{\psi_{i^{\sigma_1(1)}_1,\ldots,i^{\sigma_{k-1}(1)}_{k-1},i^{1}_k}}\cdots\overline{\psi_{i^{\sigma_1(m)}_1,\ldots,i^{\sigma_{k-1}(m)}_{k-1},i^{m}_k}}  \\
= & \sum_{i^{1}_{1},\ldots,i^{m_1}_{k}}\sum_{i^{m_1+1}_{1},\ldots,i^{m}_{k}}\psi_{i^{1}_1,\ldots,i^{1}_k}\cdots\psi_{i^{m}_1,\ldots,i^{m}_k}\overline{\psi_{i^{\sigma_1(1)}_1,\ldots,i^{\sigma_{k-1}(1)}_{k-1},i^{1}_k}}\cdots\overline{\psi_{i^{\sigma_1(m)}_1,\ldots,i^{\sigma_{k-1}(m)}_{k-1},i^{m}_k}}  \\
= & \sum_{i^{1}_{1},\ldots,i^{m_1}_{k}}\psi_{i^{1}_1,\ldots,i^{1}_k}\cdots\psi_{i^{m_1}_1,\ldots,i^{m_1}_k}\overline{\psi_{i^{\sigma^{(1)}_1(1)}_1,\ldots,i^{\sigma^{(1)}_{k-1}(1)}_{k-1},i^{1}_k}}\cdots\overline{\psi_{i^{\sigma^{(1)}_1(m_1)}_1,\ldots,i^{\sigma^{(1)}_{k-1}(m_1)}_{k-1},i^{m_1}_k}}  \\
\cdot & \sum_{i^{1}_{1},\ldots,i^{m_2}_{k}}\psi_{i^{1}_1,\ldots,i^{1}_k}\cdots\psi_{i^{m_2}_1,\ldots,i^{m_2}_k}\overline{\psi_{i^{\sigma^{(2)}_1(1)}_1,\ldots,i^{\sigma^{(2)}_{k-1}(1)}_{k-1},i^{1}_k}}\cdots\overline{\psi_{i^{\sigma^{(2)}_1(m_2)}_1,\ldots,i^{\sigma^{(2)}_{k-1}(m_2)}_{k-1},i^{m_2}_k}}  \\
\end{split}
\end{equation}
\end{proof}
In other words, the map $\bigsqcup_{m=1}^{\infty}S_m^{k-1}/S_m\to I_k$ given
by $s\mapsto f_s$ is a semigroup-homomorphism.

Now we are ready to prove our main theorem:
\begin{thm}
$I_k$ is freely generated by the set
\begin{equation}
F:=\{f_{\Phi(\tilde{G})}|\tilde{G}\in\Isoc(G)\}
\end{equation}
\end{thm}
\begin{proof}
In ref. \cite{HWW} it was shown that the set
\begin{equation}
\{f_s|s\in S_m^{k-1}/S_m\}=\{f_{\Phi(\tilde{G})}|\tilde{G}\in\Iso(G,m)\}
\end{equation}
forms a basis of the degree $m$ homogenous subspace of $I_{k,n}$ (when represented
as polynomials) if $n\ge(m,\ldots,m)$. Therefore, it is also a basis of the degree
$m$ homogenous subspace of $I_k$. As $I_k$ is the direct sum of its homogenous subspaces,
we conclude that $\{f_{\Phi(\tilde{G})}|\tilde{G}\in\Iso(G)\}$ is a basis of $I_k$.
Note that this also implies that the map $\tilde{G}\mapsto f_{\Phi(\tilde{G})}$ is injective.

An element of the form $f_s$ where $s\in S_m^{k-1}/S_m$ can be uniquely written
as the product of some elements of $F$. Indeed, $\Phi^{-1}(s)$ is a covering of
$G$, which can be uniquely written as a disjoint union of connected coverings
$\tilde{G}_1,\ldots,\tilde{G}_d$ (up to isomorphism and ordering), and therefore
\begin{equation}
f_s = f_{\Phi(\tilde{G}_1\sqcup\cdots\sqcup\tilde{G}_d)} = f_{\Phi(\tilde{G}_1)\star\cdots\star\Phi(\tilde{G}_{d})}
 = f_{\Phi(\tilde{G}_1)}\cdots f_{\Phi(\tilde{G}_{d})}
\end{equation}
\end{proof}

\section{Conclusion}\label{sec:conclusion}

We have shown that the inverse limit $I_k$ of the algebras of LU-invariant
polynomials of pure states of $k$-partite quantum systems with finite dimensional
Hilbert spaces is free, and an algebraically independent generating set can be
given in terms of finite connected coverings of a graph with a single vertex and
$k$ loops. The number of homogenous degree $2d$ polynomials in the algebraically
independent generating set equals the number of isomorphism classes of $d$-fold
connected coverings, which in turn equals the number of conjugacy classes of
index $d$ subgroups of a free group on $k-1$ generators.

In light of the close relationship between LU-equivalence classes of mixed
states of a $k$-particle quantum system and those of pure states of a $k+1$-particle
quantum system \cite{Albeverio,Vrana}, one should be able to interpret our result in the context of
mixed states. This can be done as follows.

Observe that each term on the right hand side of eq. (\ref{eq:f}) depends only
on the reduced density matrix obtained when we trace over the last subsystem.
Therefore it is easy to translate the result to the case of mixed state local
unitary invariants. Let $I^{mixed}_{k}$ denote the inverse limit of the algebras
of LU-invariants of mixed states over $k$-partite quantum systems with finite
dimensional Hilbert space as in \cite{Vrana}. Let $G$ be the graph with a single
vertex and $k$ directed labelled edges. To a connected covering $\tilde{G}\in\Iso(G,m)$
we associate the following invariant with $[\sigma_1,\ldots,\sigma_{k}]=\Phi(\tilde{G})$.
\begin{equation}
f_{[\sigma_1,\ldots,\sigma_{k}]}(\varrho)=\sum_{i^{1}_{1},\ldots,i^{m}_{k}}\varrho_{i^{1}_1,\ldots,i^{1}_k,i^{\sigma_1(1)}_1,\ldots,i^{\sigma_{k}(1)}_{k}}\cdots\varrho_{i^{m}_1,\ldots,i^{m}_k,i^{\sigma_1(m)}_1,\ldots,i^{\sigma_{k}(m)}_{k}}
\end{equation}
where
\begin{equation}
\varrho=\sum_{\substack{i_1,\ldots,i_k \\ j_1,\ldots,j_k}}\varrho_{i_1,\ldots,i_k,j_1,\ldots,j_k}e_{i_1}\otimes\cdots\otimes e_{i_k}\otimes e^*_{j_1}\otimes\cdots\otimes e^*_{j_k}
\end{equation}
is an arbitrary mixed state.

Explicite descriptions of the algebras $I_{k,n}$ are known in only a limited
number cases, including $k=2$, $n$ arbitrary, $k=3$, $n=(2,2,2)$ \cite{MW} and
$k=4$, $n=(2,2,2,2)$ \cite{Wallach}. It should be noted that for any $k\in\mathbb{N}$ and
$n\in\mathbb{N}^k$, $I_{k,n}$ is a quotient of $I_k$. It would be interesting
to determine the kernels of the quotient maps in each case.

We would like to emphasize that in spite of our lack of knowledge about the
structure of every single $I_{k,n}$, fortunately the generators of $I_k$ can be
directly interpreted as generators of the algebras of invariants of pure states
of arbitrary $k$-partite quantum systems.

\end{document}